\documentclass[conference]{IEEEtran}
\makeatletter
\def\ps@headings{%
\def\@oddhead{\mbox{}\scriptsize\rightmark \hfil \thepage}%
\def\@evenhead{\scriptsize\thepage \hfil \leftmark\mbox{}}%
\def\@oddfoot{}%
\def\@evenfoot{}}
\makeatother
\usepackage{amsmath}
\usepackage{amsfonts}
\usepackage{amssymb}
\usepackage{amsthm}
\usepackage{tikz}

\usetikzlibrary{arrows,decorations.pathmorphing,backgrounds,positioning,fit,petri}

\def\baselinestretch{1.1}

\newtheorem{property}{Property}
\newtheorem{lemma}{Lemma}

\newtheorem{theorem}{Theorem}
\newtheorem{corollary}{Corollary}

\DeclareMathOperator{\ii}{In}
\DeclareMathOperator{\oo}{Out}
\DeclareMathOperator{\mm}{Mid}

\begin{document}

\title{Exact Cooperative Regenerating Codes with Minimum-Repair-Bandwidth for Distributed Storage}
\author{\IEEEauthorblockN{Anyu~Wang and
        Zhifang~Zhang}
\IEEEauthorblockA{Key
Laboratory of Mathematics Mechanization\\Academy of Mathematics and Systems Science, CAS\\ Beijing, China\\
Email: wanganyu10@mails.gucas.ac.cn,~ zfz@amss.ac.cn}}

\maketitle
\begin{abstract}
We give an explicit construction of exact cooperative regenerating codes at the MBCR (minimum bandwidth cooperative regeneration) point. Before the paper, the only known explicit MBCR code is given with parameters $n=d+r$ and $d=k$, while our construction applies to all possible values of $n,k,d,r$. The code has a brief expression in the polynomial form and the data reconstruction is accomplished by bivariate polynomial interpolation. It is a scalar code and  operates over a finite field of size $q\geq n$. Besides, we establish several subspace properties for linear exact MBCR codes. Based on these properties we prove that linear exact MBCR codes cannot achieve repair-by-transfer.

\end{abstract}

\section{Introduction}
Distributed storage system provides a preferable solution to the requirements of
large storage volume and widespread data access. To avoid data loss
from storage node failures, erasure coding is frequently used in
distributed storage systems, such as Total Recall
\cite{B:TotalRecall} and Oceanstore \cite{Ketal:OceanStore}. It
encodes the data file into $n$ pieces, distributing to $n$ nodes
respectively in the network, and a data-collector can retrieve the
original file by connecting to any $k$ storage nodes. This process
of data retrieval is referred to as {\it data reconstruction}. When
a node fails or leaves the system, a self-sustaining storage system
should be able to repair or regenerate the node by downloading data
from survival nodes (called {\it helper nodes}). This process is called {\it node repair}, and the total amount of data downloaded during the process is referred
to as {\it repair bandwidth}. Traditional erasure codes mostly need
repair bandwidth equal to the size of the entire file, which is much
larger than the piece stored at each node. Dimakis et al.
\cite{Dimakis:RCfirst} discover a tradeoff between the node storage
and repair bandwidth. They propose a new kind of erasure codes, named
{\it regenerating codes}, which achieves the tradeoff.
Regenerating codes with minimum storage and with minimum repair
bandwidth have been constructed explicitly
\cite{Rashmi:ExplicitMBR,Rashmi:productMatrix,Shah:IA}.

Most of the studies on regenerating codes are for single-failure
recovery, while in several scenarios multiple failures need to be
considered. For example, in Total Recall a repair process is
triggered only after the total number of failed nodes has reached a
predefined threshold. Suppose $r$ newcomers are to be generated to
replace the failed nodes in a system. Comparing with the one-by-one
repair manner, {\it cooperative repair} is more profitable because
the bandwidth between the newcomers is also used. That is, each
newcomer is allowed to firstly download data from $d$ helper nodes  and
then from the other $r-1$ newcomers. The idea of cooperative repair
first appears in \cite{Hu:Cooperative} with $d=n-r$. Then paper
\cite{Wang:MFR} considers the repair with flexible $d$'s. We call
regenerating codes with cooperative repair as {\it cooperative
regenerating codes}. The tradeoff between node storage and repair
bandwidth for cooperative regenerating codes is given in
\cite{Ker:Coordinated}. Two extreme points in the tradeoff are
called MBCR (i.e. minimum bandwidth cooperative regeneration) and
MSCR (i.e. minimum storage cooperative regeneration). They meet
minimum repair bandwidth and minimum storage respectively.

There are two major repair modes in regenerating codes. One is {\it
exact repair}, namely the lost content of the failed node are
regenerated exactly. The other is {\it functional repair} which means
the content of the newcomer may not be the same as in the failed
one, but the system maintains the property of data reconstruction.
General bounds and implicit constructions of regenerating codes with
functional repair can be developed from results of network coding
\cite{WDR07:DeterminiticRC,Hu:Cooperative}. Since exact repair
brings less changes to the system than functional repair, people
cares more about explicit constructions of exact regenerating codes.
Additionally, in practice it is also desirable to  minimize the number of bits a node must read out from its memory during the repair of failed nodes. Recently people \cite{Shah:repairBYtrans,Shum&Hu:repairBYtrans} start to study the {\it repair-by-transfer} regenerating code in which the number of bits read out during the repair is minimal, namely equal to the number of bits to be sent out.

About cooperative regenerating codes, Shum \cite{Shum:MSCR} gives
an explicit construction of exact MSCR codes with parameters $d=k$, then
he and Hu \cite{Shum&Hu:MBCR} construct exact MBCR codes in the case of
$d=k$ and $n=d+r$. Recently, paper \cite{Nicolas:MSCR} constructs exact MSCR codes for $k=2$ and $d\geq k$, and shows impossibility of scalar exact MSCR codes under $k\geq3$ and $d>k$. Paper \cite{Shun&Hu:existence} proves the
existence of MBCR codes with functional repair for general parameters.

In this paper, we explicitly construct an exact MBCR code for all possible values of $n,k,d,r$. The code has a brief expression in the polynomial form and the data reconstruction is accomplished by bivariate polynomial interpolation. Moreover, the code is scalar and  operates over a finite field of size $q\geq n$. Besides, we establish several subspace properties for linear exact MBCR codes. Based on these properties we prove that linear exact MBCR codes cannot achieve repair-by-transfer.

Organization of the paper is as follows. Section 2 describes the problem of cooperative regenerating codes. Section 3 derives subspace properties of exact MBCR codes and proves the impossibility result about repair-by-transfer. Section 4 gives the explicit construction of MBCR codes and Section 5 concludes the paper.

\begin{figure*}[ht]
\begin{center}
\begin{tikzpicture}
\tiny
  \node [minimum size= 7.5mm,circle,draw,fill=blue!20] (source) at(0,2.5) {$S$};

  \node [minimum size= 7.5mm,circle,draw] (in11) at (2,4) {$\ii_1$};
  \node [minimum size= 7.5mm,circle,draw] (in12) at (2,3) {$\ii_2$};
  \node [minimum size= 7.5mm,circle,draw] (in13) at (2,2) {$\ii_3$};
  \node [minimum size= 7.5mm,circle,draw] (in14) at (2,1) {$\ii_4$};
  \node [minimum size= 7.5mm,circle,draw] (in15) at (2,0) {$\ii_5$};
  \node [minimum size= 7.5mm,circle,draw] (out11) at (3.5,4) {$\oo_1$};
  \node [minimum size= 7.5mm,circle,draw] (out12) at (3.5,3) {$\oo_2$};
  \node [minimum size= 7.5mm,circle,draw] (out13) at (3.5,2) {$\oo_3$};
  \node [minimum size= 7.5mm,circle,draw] (out14) at (3.5,1) {$\oo_4$};
  \node [minimum size= 7.5mm,circle,draw] (out15) at (3.5,0) {$\oo_5$};

  \node [minimum size= 7.5mm,circle,draw,fill=green!20] (in23) at (5.5,2) {$\ii_{6}$};
  \node [minimum size= 7.5mm,circle,draw,fill=green!20] (in24) at (5.5,1) {$\ii_{7}$};
  \node [minimum size= 7.5mm,circle,draw,fill=green!20] (mid23) at (7,2) {$\mm_{6}$};
  \node [minimum size= 7.5mm,circle,draw,fill=green!20] (mid24) at (7,1) {$\mm_{7}$};
  \node [minimum size= 7.5mm,circle,draw,fill=green!20] (out23) at (8.5,2) {$\oo_{6}$};
  \node [minimum size= 7.5mm,circle,draw,fill=green!20] (out24) at (8.5,1) {$\oo_{7}$};

  \node [minimum size= 7.5mm,circle,draw,fill=blue!20] (dc1) at (10,3.7) {DC};
  \node [minimum size= 7.5mm,circle,draw,fill=blue!20] (dc2) at (10,2.7) {DC};
  \node [minimum size= 7.5mm,circle,draw,fill=blue!20] (dc3) at (10,0) {DC};

  \draw [thick,color=red] (2,2.2) -- (3.5,1.8);
  \draw [thick,color=red] (2,1.8) -- (3.5,2.2);

  \draw [thick,color=red] (2,1.2) -- (3.5,0.8);
  \draw [thick,color=red] (2,0.8) -- (3.5,1.2);

  \node at (10,1.7) {$\bullet$};
  \node at (10,1.4) {$\bullet$};
  \node at (10,1.1) {$\bullet$};

  \draw [->] (source) to node [auto] {$\infty$} (in11);
  \draw [->] (source) to node [auto] {$\infty$} (in12);
  \draw [->] (source) to node [auto] {$\infty$} (in13);
  \draw [->] (source) to node [auto] {$\infty$} (in14);
  \draw [->] (source) to node [auto] {$\infty$} (in15);

  \draw [->] (in11) to node [auto] {$\alpha$} (out11);
  \draw [->] (in12) to node [auto] {$\alpha$} (out12);
  \draw [->] (in13) to node [auto] {$\alpha$} (out13);
  \draw [->] (in14) to node [auto] {$\alpha$} (out14);
  \draw [->] (in15) to node [auto] {$\alpha$} (out15);

  \draw [->] (out11) to node [auto] {$\beta_1$} (in23);
  \draw [->] (out12) to node [auto] {$\beta_1$} (in23);
  \draw [->] (out15) to node [auto] {$\beta_1$} (in23);
  \draw [->] (out11) to node [auto][swap] {$\beta_1$} (in24);
  \draw [->] (out12) to node [auto][swap] {$\beta_1$} (in24);
  \draw [->] (out15) to node [auto][swap] {$\beta_1$} (in24);

  \draw [->] (in23) to node [auto] {$\infty$} (mid23);
  \draw [->] (in24) to node [auto] {$\infty$} (mid24);
  \draw [->] (in23) to node [auto] {$\beta_2$} (mid24);
  \draw [->] (in24) to node [auto] {$\beta_2$} (mid23);

  \draw [->] (mid23) to node [auto] {$\alpha$} (out23);
  \draw [->] (mid24) to node [auto] {$\alpha$} (out24);

  \draw [->] (out11) to node [auto] {$\infty$} (dc1);
  \draw [->] (out23) to node [auto] {$\infty$} (dc1);

  \draw [->] (out12) to node [auto] {$\infty$} (dc2);
  \draw [->] (out23) to node [auto] {$\infty$} (dc2);

  \draw [->] (out24) to node [auto] {$\infty$} (dc3);
  \draw [->] (out15) to node [auto] {$\infty$} (dc3);

\normalsize
  \node at (0,-1) {Stage -1};
  \node at (2.75,-1) {Stage 0};
  \node at (7,-1) {Stage 1};
\end{tikzpicture}
\caption{An information flow graph of cooperative regenerating code ($n=5,k=2,d=3,r=2$). }\label{fig1}
\end{center}
\end{figure*}
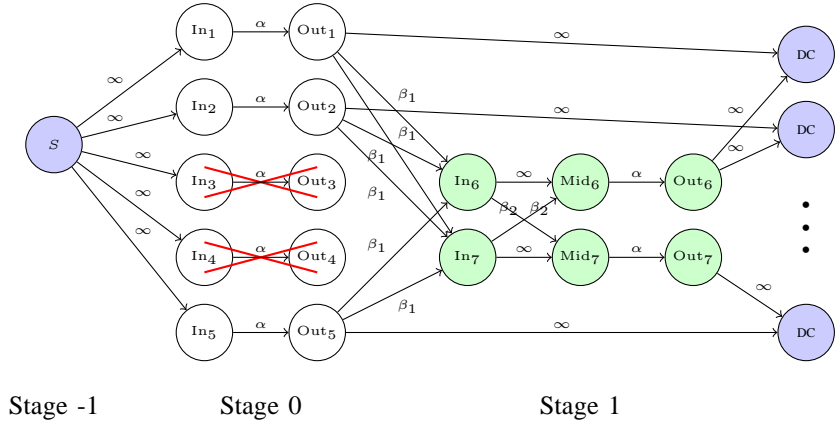

\section{Problem Description}
As in \cite{Shum&Hu:MBCR}, we describe the problem of cooperative
regenerating code in stages and give the corresponding information
flow graph.
\begin{itemize}
\item In stage $-1$, a source vertex $S$ holds the original
data file consisting of $B$ packets.

\item In stage $0$, the encoded file is distributed to $n$ nodes,
each storing $\alpha$ packets. To make the storage clear in the
information flow graph, we split each node $i\in\{1,...,n\}$ into
two nodes $\textsf{In}_i$ and $\textsf{Out}_i$ with a directed edge
of capacity $\alpha$ from $\textsf{In}_i$ to $\textsf{Out}_i$.

\item For $i=1,2,...$, stage $i$ is triggered at the failure of $r$ nodes. Then $r$ newcomers are generated to replace the failed nodes through two phases: firstly, each newcomer connects to $d$ survival nodes (called helper nodes) and downloads $\beta_1$ packets from each; secondly, it downloads $\beta_2$ packets from each of the other $r-1$ newcomers. Similarly, we split each newcomer into three nodes $\textsf{In}_i,\;\textsf{Mid}_i$ and $\textsf{Out}_i$ in the information flow graph.

\item Data-collector DC connecting to any $k$ active nodes can recover the original data file, as required by the {\it data reconstruction} property.
\end{itemize}

Obviously, the parameters should satisfy $d+r\leq n$, $1\leq k\leq n$, $\beta_1\leq \alpha$, etc. Note that if $ d < k$, a data collector can reconstruct the data file by connecting to any $d$ nodes since any set of failed nodes can be regenerated by these $d$ nodes. Thus, a $( n, k, d, r )$ cooperative regenerating code implies a $( n, k' =d, d, r )$ code and vice versa.
Without loss of generality we assume $d \ge k$ throughout the paper.

Figure \ref{fig1} displays an information flow graph for the cooperative regenerating code with parameters $(n=5,k=2,d=3,r=2)$. The labels $\alpha,\beta_1,\beta_2,\infty$ denote the capacity of the corresponding edges. Thus the problem of cooperative regenerating codes induces a multicast problem in such a graph where $S$ is the single source and all possible DC's are the sinks. Furthermore, this graph illustrates a specifical fail-repair process. There are infinitely many fail-repair processes since the node failures and edge links are both variable. Each process gives an information flow graph. Therefore a cooperative regenerating code with parameters $(n,k,d,r,\alpha,\beta_1,\beta_2)$ implies a multicast coding in all the graphs. As a result, the cut-set bound for single-source multicast problem \cite{ACLY:NetCodFirst} gives the following necessary condition for cooperative regenerating code \cite{Hu:Cooperative,Ker:Coordinated,Shum:MSCR}.
\begin{equation}B\le \sum_{h=1}^s l_h \min\{\alpha,(d-\sum_{t=1}^{h-1} l_t) \beta_1+(r-l_h) \beta_2\}\label{eq1}\end{equation}
where $\{l_h\}_{h=1}^s$ is any set of integers satisfying $l_1+\cdots+l_s=k$ and $1\leq l_1,\cdots,l_s\leq r$. Actually, $l_i$ means the data-collector connects to $l_i$ nodes which join the system from stage $i$ and remain active thereafter.

From  bound (\ref{eq1}) it can see there is a tradeoff between node storage $\alpha$ and repair bandwidth $\gamma=d\beta_1+(r-1)\beta_2$. The MBCR point is an extreme point on the tradeoff which has the minimum repair bandwidth. Specifically, it has the parameters \cite{Ker:Coordinated}:
\begin{equation}\alpha=\gamma,\;\;\beta_1=2\beta_2,\;\;\beta_2=\frac{B}{k(2d+r-k)}\;.\label{eq2}\end{equation}
Another extreme point is MSCR with parameters
$$\alpha=\frac{B}{k},\;\;\beta_1=\beta_2=\frac{B}{k(d-k+r)}\;.$$ We focus on MBCR codes in this paper.

However bound (\ref{eq1}) is deduced for functional repair, it is still unknown if this bound is tight for exact cooperative regenerating codes. Explicit constructions of exact MSCR codes and MBCR codes have been given only for special parameters \cite{Shum:MSCR,Shum&Hu:MBCR,Nicolas:MSCR}. In the paper, we explicitly construct an exact MBCR code for all possible values of $n,k,d,r$, which means bound (\ref{eq1}) can be met for exact cooperative regenerating codes at the MBCR point.

\section{Subspace Properties of Exact  MBCR Codes}
We first introduce some notations and review some basic results about linear subspaces.

Consider a linear exact MBCR codes with parameters $(n,k,d,r,\alpha,\beta_1,\beta_2)$. Suppose each packet is an element in a finite field $\mathbb{F}_q$. Then the original data file can be seemed as a vector $u\in\mathbb{F}_q^B$. For consistence we assume the vectors throughout this paper are column vectors. Since the code is linear, each node $i\in\{1,...,n\}$ stores $\alpha$ packets which are linear combinations of the original data packets. Specifically, suppose node $i$ stores $u^\tau g_1^{(i)}, u^\tau g_2^{(i)}, \cdots, u^\tau g_\alpha^{(i)}$, where $g_j^{(i)}\in\mathbb{F}_q^B$ are predetermined for $1\leq j\leq \alpha$. Linear operations performed on the stored packets correspond to the same operations performed on the vectors $g_j^{(i)},\;1\leq j\leq \alpha$. Hence we say node $i$ stores a subspace $W_i$ spanned by $g^{(i)}_1,\cdots,g^{(i)}_{\alpha}$. Similarly, when node $i$ passes packets $u^\tau g_{i_1}^{(i)},...,u^\tau g_{i_{\beta_1}}^{(i)}$ to another node, we say the subspace spanned by $g_{i_1}^{(i)},...,g_{i_{\beta_1}}^{(i)}$ is transferred.

Suppose $R$ is a set of $r$ failed nodes. For $i\in R$, let $\mathcal{H}_R^{(i)}$ denote the set of $d$ helper nodes that each provides $\beta_1$ packets to help repair node $i$. For $i,i'\in R$ and $j\in\mathcal{H}_R^{(i)}$, let $S^{j,i}_{R}$ be the subspace passed from  $j$ to $i$ and $T^{i,i'}_{R}$ the subspace passed from $i$ to $i'$. That is, $S^{j,i}_{R}$ is contribution of helper nodes in the repair process and $T^{i,i'}_{R}$ is exchange between the newcomers. Note that $S^{j,i}_{R}$ and $T^{i,i'}_{R}$ also depend on $\{\mathcal{H}_R^{(l)}\mid l\in R\}$. For simplicity, we fix $\{\mathcal{H}_R^{(l)}\mid l\in R\}$ for each $R$. Thus subspaces with subscript $R$ are always defined under the same $\{\mathcal{H}_R^{(l)}\mid l\in R\}$. obviously, we have $\dim \{W_i\} \le \alpha,\;\dim \{S^{j,i}_{R}\} \le \beta_1$ and $\dim \{T^{i,i^\prime}_{R}\} \le \beta_2$. Furthermore, since the repair is exact, the subspaces $W_i,S^{j,i}_{R},T^{i,i'}_{R}$ keep invariant.

Let $E_1,E_2$ be two subspaces of $\mathbb{F}_q^B$, their sum is  defined by $E_1 + E_2 = \{ e_1 +e_2 | e_1 \in E_1, e_2 \in E_2 \}$.
If $E_1 \cap E_2$ contains only zero vector, $E_1 + E_2$ is called the direct sum of $E_1$ and $E_2$, denoted by $E_1 \oplus E_2$. For $m$ subspaces $E_1,...,E_m\subseteq\mathbb{F}_q^B$, define $\oplus_{i=1}^mE_i=E_1\oplus(\oplus_{i=2}^mE_i)$. The following theorem is a well known result in linear algebra.

\begin{theorem}\label{th1}
Let $E_1,...,E_m$ be $m$ subspaces of $\mathbb{F}_q^B$. The following statements are equivalent:
\begin{itemize}
\item [$(a)$] $\sum_{i=1}^mE_i=\oplus_{i=1}^mE_i$.

\item [$(b)$] $\dim \{\sum_{i=1}^mE_i\}=\sum_{i=1}^m\dim\{E_i\}$.

\item [$(c)$] $E_i\bigcap(\sum_{j\neq i}E_j)=\{0\}$.
\end{itemize}
\end{theorem}
Now we are ready to investigate subspace properties of linear exact MBCR codes.

\begin{lemma}\label{le1}
Suppose $I\subseteq R$ and $J\subseteq \bigcap_{i\in I} \mathcal{H}_R^{(i)}$. Moreover, $|I|=a$ and $|J|=b$.
Then
$$ \dim \{ \sum_{i\in I} W_{i} \} - \dim \{ (\sum_{i\in I} W_{i}) \cap (\sum_{j\in J} W_{j}) \}$$$$\le a ( (d-b) \beta_1 + (r-a) \beta_2 )\;.$$
\end{lemma}
\begin{proof}
Denote $ \tilde{R} = R \backslash I$ and $\tilde{\mathcal{H}}_R^{(i)} = \mathcal{H}_R^{(i)} \backslash J$. Because a failed node can be repaired through two phases, for all $i\in R$ it has
\begin{equation}\label{eq33} W_{i} \subseteq \sum_{j \in \mathcal{H}_R^{(i)}} S_{R}^{j,i} + \sum_{i' \in R \backslash \{ i\}} T_{R}^{i',i}\;.\end{equation}
Thus
\begin{eqnarray*}
& &\sum_{i\in I} W_{i}
\subseteq  \sum_{i\in I} ( \sum_{j \in \mathcal{H}_R^{(i)}} S_{R}^{j,i} + \sum_{i' \in R \backslash \{ i \}} T_{R}^{i',i} ) \\
& = & \sum_{i\in I} ( \sum_{j \in \mathcal{H}_R^{(i)}} S_{R}^{j,i} + \sum_{i' \in \tilde{R} } T_{R}^{i',i} + \sum_{ i' \in I \backslash \{ i \} } T_{R}^{i',i} ) \\
& \stackrel{(a)}{\subseteq} & \sum_{i\in I} ( \sum_{j \in \mathcal{H}_R^{(i)}} S_{ R}^{j,i} + \sum_{i' \in \tilde{R} } T_{R}^{i',i} + \sum_{ i' \in I \setminus \{ i\}\atop j \in \mathcal{H}_R^{(i')}} S_{R}^{j,{i'}} ) \\& = & \sum_{i\in I} ( \sum_{j \in \mathcal{H}_R^{(i)}} S_{R}^{j,i} + \sum_{i' \in \tilde{R} } T_{R}^{i',i} ) + \sum_{ i \in I} \sum_{j \in \mathcal{H}_R^{(i)}} S_{R}^{j,i} \\
& = & \sum_{i\in I} ( \sum_{j \in \mathcal{H}_R^{(i)}} S_{R}^{j,i} + \sum_{i' \in \tilde{R} } T_{R}^{i',i} ) \\
& = & \sum_{i\in I} ( \sum_{j \in J} S_{R}^{j,i} + \sum_{j \in \tilde{\mathcal{H}}_R^{(i)}} S_{R}^{j,i} + \sum_{i' \in \tilde{R} } T_{R}^{i',i} ) \\
& \subseteq & \sum_{j\in J} W_{j} + \sum_{i\in I}( \sum_{ j \in \tilde{\mathcal{H}}_R^{(i)}} S_{R}^{j,i} + \sum_{i' \in \tilde{R}} T_{ R}^{i',i} ),
\end{eqnarray*}
where $(a)$ follows from $ T_{R}^{i',i} \subseteq \sum_{j \in \mathcal{H}_R^{(i')}} S_{R}^{j,i'}$, since the packets passed by node $i' $ to $i$ in the second repair phase are linear combinations of the packets it received in the first phase.

Therefore,
\begin{eqnarray*}
& & \dim \{ \sum_{i\in I} W_{i} \} - \dim \{ (\sum_{i\in I} W_{i}) \cap (\sum_{j\in J} W_{j}) \} \\
& = & \dim\{ \sum_{i\in I} W_{i} + \sum_{j\in J} W_{j} \} - \dim\{ \sum_{j\in J} W_{j} \} \\
& \le & \dim \{ \sum_{j\in J} W_{j} + \sum_{i\in I}( \sum_{ j \in \tilde{\mathcal{H}}_R^{(i)}} S_{R}^{j,i} + \sum_{i' \in \tilde{R}} T_{ R}^{i',i} ) \}\\\;\;&& -\dim \{ \sum_{j\in J} W_{j} \} \\
& \le & \dim \{ \sum_{i\in I} ( \sum_{ j \in \tilde{\mathcal{H}}_R^{(i)}} S_{ R}^{j,i} + \sum_{i' \in \tilde{R}} T_{R}^{i',i} ) \} \\
& \le & a((d-b) \beta_1 + (r-a) \beta_2)
\end{eqnarray*}
\end{proof}

The above lemma provides a fundamental result for proving the subspace properties.
Actually it holds for all linear exact cooperative regenerating codes, although we use it only for exact MBCR codes in the following.

\begin{property}\label{pro1}
For $1\leq i\neq j\leq n$, $\dim\{W_i\}=\alpha$, and
$\dim\{W_i\cap W_j\}=\beta_1$.
\end{property}
\begin{proof}
Without loss of generality, we prove that $\dim \{W_1\} = \alpha$, $\dim \{W_1 \cap W_2\} = \beta_1$.

Consider a particular fail-repair process where a data-collector connects to node $1,...k$, and for $1\leq i\leq k$ node $i$ is regenerated at the $i$-th stage and remains active thereafter. Moreover, node $i$ help repair node $j$ for all $1\leq i<j\leq k$, i.e., $\{1,...,i-1\}\subset \mathcal{H}_{R_i}^{(i)}$ for $1<i\leq k$, where $R_i$ is the set of failed nodes at the $i$-th stage. Since the data reconstruction property is held for any fail-repair process, we have $\mathbb{F}^{B} \subseteq W_1 + \cdots + W_k$ , which implies
\begin{equation}\label{eq3}B \le \dim \{W_1 + \cdots + W_k\}\;.\end{equation}
On the other hand,
\begin{eqnarray*}
&&\dim \{W_1 + \cdots + W_k\}\\&=&\dim \{W_1\}+\sum_{i=2}^k (\dim\{\sum_{j=1}^iW_j\}-\dim\{\sum_{j=1}^{i-1}W_j\})\\&=&\dim \{W_1\}+\sum_{i=2}^k(\dim \{W_i\}-\dim\{W_i\cap\sum_{j=1}^{i-1}W_j\})\\&\stackrel{(a)}{\le}&\alpha+\sum_{i=2}^k((d-i+1)\beta_1+(r-1)\beta_2)
\\&\stackrel{(b)}{=}&B\;,
\end{eqnarray*}
where (a) is from Lemma \ref{le1} and (b) from parameters of MBCR displayed in (\ref{eq2}). Because of (\ref{eq3}), (a) must hold with equality. Namely, $\dim\{ W_1\}=\alpha$ and \begin{equation}\label{eq4}\dim \{W_i\}-\dim\{W_i\cap\sum_{j=1}^{i-1}W_j\}=(d-i+1)\beta_1+(r-1)\beta_2\end{equation} for $2\leq i\leq k$. Thus we have proven $\dim \{W_1\}=\alpha$. A similar proof states $\dim \{W_i\}=\alpha$ for all $i$.

Fix $i=2$ in (\ref{eq4}), it follows $\dim \{W_2\}-\dim\{W_2\cap W_1\}=(d-1)\beta_1+(r-1)\beta_2$. Since $\dim \{W_2\}=\alpha$ and $\alpha=d\beta_1+(r-1)\beta_2$ for MBCR codes, we get $\dim \{W_1 \cap W_2\} = \beta_1$.
\end{proof}

\begin{property}\label{pro2}
For all $i \in R$, $$W_i = (\bigoplus_{j \in \mathcal{H}_R^{(i)}} S_{R} ^{ j, i}) \oplus (\bigoplus_{ i' \in R \atop i' \neq i} T_{R} ^{ i',i }).$$
\end{property}
\begin{proof}
As stated in (\ref{eq33}), $$ W_i \subseteq \sum_{j \in \mathcal{H}_R^{(i)}} S_{ R} ^{ j, i} + \sum_{ i^\prime \in R \atop i^\prime \neq i} T_{ R} ^{ i^\prime,i }\;.$$ Considering the dimensions of the two sides, we have
\begin{eqnarray*}
\alpha & \le & \dim \{ \sum_{j \in \mathcal{H}_R^{(i)}} S_{R} ^{ j, i} + \sum_{ i^\prime \in R \atop i^\prime \neq i} T_{ R} ^{ i^\prime,i } \} \\
       & {\le} & \sum_{j \in \mathcal{H}_R^{(i)}} \dim \{S_{R} ^{ j, i} \}+ \sum_{ i^\prime \in R \atop i^\prime \neq i} \dim \{T_{R} ^{ i^\prime,i } \}\\
       & \le & d \beta_1 + (r-1) \beta_2 \\
       & = & \alpha
\end{eqnarray*}
where the last equality comes from parameters of MBCR codes. Therefore, all the equalities above must hold. Then by Theorem \ref{th1}  the property is proved.
\end{proof}

\begin{corollary}\label{coro1}
For all $i,i'\in R, i'\neq i$ and $ j \in \mathcal{H}_R^{(i)}$, it has $\dim \{ S^{j,i}_{R} \} = \beta_1$ and $\dim \{ T^{i^\prime,i}_{R} \} = \beta_2$.
\end{corollary}

\begin{property}\label{pro3}
For all $i,i'\in R, i'\neq i$ and $ j \in \mathcal{H}_i$, it has $S^{j,i}_{R} = W_i \cap W_j$ and $T^{i,i^\prime}_{R} \oplus T^{i^\prime,i}_{R} = W_i \cap W_{i^\prime}$.
\end{property}
\begin{proof}
We have $S^{j,i}_{R}\subseteq W_i$ from Property \ref{pro2} and $S^{j,i}_{R}\subseteq W_j$ from the definition of $S^{j,i}_{R}$. Thus $S^{j,i}_{R} \subseteq W_i \cap W_j$. Similarly, $ T^{i^\prime,i}_{R}, T^{i,i'}_{R}\subseteq W_i \cap W_{i^\prime}$. Thus $T^{i^\prime,i}_{R}+ T^{i,i'}_{R}\subseteq W_i \cap W_{i^\prime}$

Form Property \ref{pro1} and Corollary \ref{coro1} we know $S^{j,i}_{R}$ and $W_i \cap W_j$ are both of dimension $\beta_1$. Hence $S^{j,i}_{R} = W_i \cap W_j$.

And,
$$ T^{i,i^\prime}_{R} \cap T^{i^\prime,i}_{R} \subseteq (\sum_{j \in \mathcal{H}_R^{(i)}} S^{j,i}_{R} ) \cap T^{i^\prime,i}_{R} = \{0\}\;,$$
where the last equality comes from Property \ref{pro2}.
Therefore $ T^{i,i^\prime}_{R} \oplus T^{i^\prime,i}_{R} \subseteq W_i \cap W_{i^\prime} $. The left side has dimension $2\beta_2=\beta_1$ from Corollary \ref{coro1} and parameters in (\ref{eq2}) for MBCR point, while the right side has dimension $\beta_1$ from Property \ref{pro1}. Hence $T^{i,i^\prime}_{R} \oplus T^{i^\prime,i}_{R} = W_i \cap W_{i^\prime}$.
\end{proof}

\subsection{Impossibility of exact repair-by-transfer}
In \cite{Rashmi:ExplicitMBR}, it studies the subspace properties of exact regenerating codes with minimum repair bandwidth and gives an explicit code in the case of $n=d+1$. The code can be seemed as a direct construction from the properties. Its significance also relies on the repair-by-transfer mode. In the following we show impossibility of exact repair-by-transfer codes at the MBCR point.

For cooperative regenerating code, repair-by-transfer is required at the first phase of the repair process. That is, in the first phase each helper node directly transfers $\beta_1$ packets it stores to the newcomer. Our impossibility result is based on the subspace properties we derived above.

\begin{theorem}
When $r\geq 2$ and $d\geq 2$, there does not exist a linear exact MBCR code that achieves repair-by-transfer.
\end{theorem}
\begin{proof}
On the contrary, we assume there is a $(n,k,d,r, \alpha,\beta_1,\beta_2)$ linear exact MBCR code that achieves repair-by-transfer. For any data file $u\in\mathbb{F}_q^B$, suppose node $1$ stores $u^\tau g^{(1)}_1,...,u^\tau g^{(1)}_\alpha$, where $g^{(1)}_1,...,g^{(1)}_\alpha$ are linearly independent vectors in $\mathbb{F}_q^B$. Denote $G=\{g^{(1)}_1,...,g^{(1)}_\alpha\}$.

For $2\leq i\leq n$, let $R_i$ be a set of $r$ failed nodes such that $i\in R_i$ and $1\in\mathcal{H}_{R_i}^{(i)}$. Suppose node $1$ transfers $u^\tau g^{(1)}_{i_1},...,u^\tau g^{(1)}_{i_{\beta_1}}$ to node $i$ in repairing $R_i$. Denote $G_i=\{g^{(1)}_{i_1},...,g^{(1)}_{i_{\beta_1}}\}$. From the definition of repair-by-transfer, $G_i\subset G$.  It is obvious that $\bigcup_{i=2}^nG_i\subseteq G$.

For any $i,j\in\{2,...,n\}$, let $R_{i,j}$ be a set of $r$ failed nodes such that $1\in R_{i,j}$ and $\{i,j\}\subseteq \mathcal{H}_{R_{i,j}}^{(1)}$. Then
\begin{eqnarray*}
G_i\cap G_j&\subset& S^{1,i}_{R_i} \cap S^{1,j}_{R_j}\\&=&(W_1\cap W_i)\cap(W_1\cap W_j)\\&=&S_{R_{i,j}}^{i,1}\cap S_{R_{i,j}}^{j,1}\\&=&\{0\}\;,
\end{eqnarray*}
where the relation $\subset$ holds because $S^{1,i}_{R_i}=span\{G_i\}$, the first two equalities come from Property \ref{pro3}, and the last equality is from Property \ref{pro2}. Since $G_i$ and $G_j$ contain only nonzero vectors, it must hold $G_i\cap G_j=\emptyset$.

Therefore $|G|\geq |\bigcup_{i=2}^nG_i|=\sum_{i=2}^n|G_i|$. Since $S^{1,i}_{R_i}=span\{G_i\}$ for $2\leq i\leq n$, from Corollary \ref{coro1} it has $|G_i|=\beta_1$. Thus $|G|\geq (n-1)\beta_1\geq (d+r-1)\beta_1>\alpha$, where the last $>$ is from $\beta_1=2\beta_2>0$ and $\alpha=d\beta_1+(r-1)\beta_2$ for MBCR codes. On the other hand, from Property \ref{pro1} it has $|G|=\alpha$. Hence we get a contradiction.
\end{proof}

The condition $r\geq 2$ is trivial for multiple node failures, and $d\geq 2$ is necessary to guarantee the repair bandwidth $\gamma<B$. Thus the above theorem proves there is no non-trivial linear exact MBCR codes which achieves repair-by-transfer.

\section{Explicit Construction of MBCR Codes}
We consider the scalar MBCR code,  i.e., $\beta_2=1$. Then according to (\ref{eq2}) it has parameters $\beta_1=2\beta_2=2$, $\alpha=d\beta_1+(r-1)\beta_2=2d+r-1$, and $B=k(2d+r-k)$. Note that our construction applies to all positive integers of $(n,k,d,r)$ such that $d+r\leq n$ and $d\geq k$.

For a data file $u\in\mathbb{F}_q^B$, we construct a bivariate polynomial over $\mathbb{F}_q$, denoted by
\begin{eqnarray}F(X,Y)=\sum_{0\leq i<k\atop 0\leq j<k}a_{ij}X^iY^j&+&\sum_{0\leq i<k\atop k\leq j<d+r}b_{ij}X^iY^j\nonumber\\&+&\sum_{k\leq i<d\atop 0\leq j<k}c_{ij}X^iY^j,\label{eq6}\end{eqnarray}
such that the $B$ components of $u$ are just its coefficients. Note $F(X,Y)$ has $k^2+k(d+r-k)+k(d-k)=k(2d+r-k)=B$ coefficients.

Then fix $n$ distinct elements $x_1,...,x_n$ in $\mathbb{F}_q$, and similarly fix distinct $y_1,...,y_n$ in $\mathbb{F}_q$. Note that it is allowed $x_i=y_j$ for some $1\leq i,j\leq n$. Thus about the field size we only require $q\geq n$.

For each node $i\in\{1,...,n\}$, it stores the values of $F(X,Y)$ at $\alpha$ points, i.e.,
$$F(x_i,y_i),\;F(x_i,y_{i\oplus1}),\;...,\;F(x_i,y_{i\oplus(d+r-1)}),$$
$$F(x_{i\oplus1},y_i),\;F(x_{i\oplus2},y_i),\;...,\;F(x_{i\oplus(d-1)},y_i),$$
where $\oplus$ denotes addition modulo $n$.
Actually, the first $d+r$ values determine the univariate polynomial $f_i(Y)=F(x_i,Y)$, since $f_i(Y)$ is of degree less than $d+r$ and can be derived from interpolation at $d+r$ distinct points. Similarly, the first value and the last $d-1$ values determine the univariate polynomial $g_i(X)=F(X,y_i)$. Therefore, we also say node $i$ stores two univariate polynomials $f_i(Y)$ and $g_i(X)$.

The validity of the above code  as an exact regenerating code for the MBCR point is established in two aspects.

(1) {\it Exact Cooperative Regeneration: }
Without loss of generality, suppose node $1,...,r$ fail and newcomers, also named node $1,...,r$ for simplicity, are to replace the failed nodes by the repair process.

In the first phase, each node $i\in\{1,...,r\}$ connects to $d$ survival nodes and downloads $\beta_1=2$ packets from each. Specifically, suppose $i$ connects to nodes $\{i_1,...,i_d\}\subseteq \{1,...,n\}\setminus\{1,...,r\}$. Then node $i_j$ sends $(F(x_{i_j},y_i),F(x_i,y_{i_j}))\in\mathbb{F}_q^2$ to $i$ for $1\leq j\leq d$. Note that node $i_j$ actually stores polynomials $f_{i_j}(Y)$ and $g_{i_j}(X)$, so it can compute $(F(x_{i_j},y_i),F(x_i,y_{i_j}))$$=(f_{i_j}(y_i),g_{i_j}(x_i))$.

Upon receiving $F(x_{i_1},y_i), F(x_{i_2},y_i),...,F(x_{i_d},y_i)$, node $i$ can get $g_i(X)=F(X,y_i)$ by the Lagrange interpolation formula, since $g_i(X)$ is of degree less than $d$. Note that node $i$ also receives $F(x_i,y_{i_1}),...,F(x_i,y_{i_d})$ and these will be used later.

In the second phase, each node $i\in\{1,...,r\}$ connects to the other $r-1$ nodes, i.e., $\{1,...,r\}\setminus\{i\}$, and downloads $\beta_2=1$ packets from each. Specifically, for $j\in\{1,...,r\}\setminus\{i\}$, node $j$ sends $F(x_i,y_j)$ to node $i$. Node $j$ can do this because it has recovered $g_j(X)$ in the first phase. Additionally, each node $i$ can compute $F(x_i,y_i)=g_i(x_i)$ by itself.

Now node $i$ has obtained $F(x_i,y_1), ..., F(x_i,y_r)$ in the second phase,  along with $F(x_i,y_{i_1}),...,$ $ F(x_i,y_{i_d})$ it received in the first phase, it can recover $f_i(Y)=F(x_i,Y)$ by interpolation.

Thus node $i$ recovers $f_i(Y)$ and $g_i(X)$, and so is exactly regenerated.

(2) {\it Data Reconstruction: }Suppose a data-collector connects to nodes $\{i_1,...,i_k\}$ to retrieve the original data file. It is equivalent to recover the polynomial $F(X,Y)$ from $\{f_{i_l}(Y),g_{i_l}(X)\mid 1\leq l\leq k\}$.

Denote \begin{eqnarray*}F(X,Y)&=&\tilde{F}(X,Y)+\sum_{j=k}^{d+r-1}(\sum_{i=0}^{k-1}b_{ij}X^i)Y^j
\\&=&\tilde{F}(X,Y)+\sum_{j=k}^{d+r-1}B_j(X)Y^j\;.\end{eqnarray*}
It can see in $\tilde{F}(X,Y)$ the degree of $Y$ is less than $k$ and for $k\leq j\leq d+r-1$ the coefficient of $Y^j$, $B_j(X)=\sum_{i=0}^{k-1}b_{ij}X^i$, is a polynomial of degree less than $k$.
For $1\leq l\leq k$, suppose $$f_{i_l}(Y)=f_0^{(i_l)}+f_1^{(i_l)}Y+\cdots+f_{d+r-1}^{(i_l)}Y^{d+r-1}\;.$$
Then for $k\leq j\leq d+r$, comparing the coefficient of $Y^j$ in $F(X,Y)$ and that in $f_{i_l}(Y)$, we get $B_j(x_{i_l})=f_j^{(i_l)}$. That is, we get the evaluation of $B_j(X)$ at $k$ distinct points $x_{i_1},...,x_{i_k}$. So for $k\leq j\leq d+r-1$, $B_j(X)$ can be recovered by interpolation, corresponding to the $b_{ij}, 0\leq i<k,k\leq j<d+r$, in (\ref{eq6}) are obtained.

Similarly, we can get $c_{ij}, k\leq i<d, 0\leq j<k$. Based on $b_{ij}$'s and $c_{ij}$'s we can further get $a_{ij}$'s in a similar way. Thus the polynomial $F(X,Y)$ is recovered, which gives the original data file.

\subsection{Subspace properties of the code}
Although it is more convenient to describe the above code in a polynomial form, we transform it into a traditional linear code to verify the subspace properties proved in Section 3.

Without loss of generality, we investigate the subspace stored by node $1$. By using the notations above, node $1$ stores a subspace spanned by:
$$\begin{array}{lllll}
(1,y_1,...&,y_1^{d+r-1}&,x_1,...&,x_1^{d-1}&,x_1y_1,...),\\
(1,y_2,...&,y_2^{d+r-1}&,x_1,...&,x_1^{d-1}&,x_1y_2,...),\\
&&\vdots&\\
(1,y_{d+r},...&,y_{d+r}^{d+r-1}&,x_1,...&,x_1^{d-1}&,x_1y_{d+r},...),\\
(1,y_1,...&,y_1^{d+r-1}&,x_2,...&,x_2^{d-1}&,x_2y_1,...),\\
&&\vdots&\\
(1,y_1,...&,y_1^{d+r-1}&,x_{d-1},...&,x_{d-1}^{d-1}&,x_{d-1}y_1,...).
\end{array}$$

That is, the first $d+r$ components of these vectors correspond to the monomials $u_{0j}Y^j$ in $F(X,Y)$ for $0\leq j<d+r$, the next $d-1$ components correspond to $u_{i0}X^i$ for $1\leq i<d$, and the remain components correspond to $u_{ij}X^iY^j$ for $i>0$ and $j>0$. Obviously, the above vectors are linearly independent, so $\dim \{W_1\}=2d+r-1=\alpha$ as proved in Property \ref{pro1}.

For any two nodes $i$ and $j$, the intersection of their spaces is spanned by
$$(1,y_i,...,y_i^{d+r-1},x_j,...,x_j^{d-1},x_jy_i,...),$$
$$(1,y_j,...,y_j^{d+r-1},x_i,...,x_i^{d-1},x_iy_j,...).$$
Correspondingly, in the repair process where $i\in R$ and $j\in\mathcal{H}_R^{(i)}$, we can see node $j$ sends $(F(x_i,y_j), F(x_j,y_i))$ to $i$, in accordance with $\dim \{W_i\cap W_j\}=\beta_1=2$ and $S^{j,i}_{R}=W_i\cap W_j$.

For another node $i'\in R$, we can see in the second repair phase, $i'$ sends $g_{i'}(x_i)=F(x_i,y_{i'})$ to $i$. The corresponding subspace is spanned by
$$(1,y_{i'},...,y_{i'}^{d+r-1},x_i,...,x_i^{d-1},x_iy_{i'},...).$$ Thus
$\dim\{T^{i',i}_{R}\}=\beta_1=1$ and $T^{i',i}_{R}\oplus T^{i,i'}_{R}=W_i\cap W_{i'}$. Based on above observations, it is also easy to verify Property \ref{pro2}.

\section{Conclusion}

We explicitly construct exact MBCR codes for all possible values of $n,k,d,r$, which can be seemed as a counterpart of the result in regenerating codes for single-failure recovery \cite{Rashmi:productMatrix}, i.e.,
explicit constructions of MBR (minimum repair-bandwidth regeneration) codes has been given for all $n,k,d$. Our code is expressed in the polynomial form and the data reconstruction is accomplished by bivariate polynomial interpolation. We note some previously given explicit constructions \cite{Rashmi:productMatrix} can also be transformed into  polynomial forms. Polynomials are expected to do more in regenerating codes.

\end{document}